\documentclass[sigconf]{acmart}
\AtBeginDocument{%
  }

\setcopyright{acmlicensed}
\copyrightyear{2018}
\acmYear{2018}
\acmDOI{XXXXXXX.XXXXXXX}
\acmConference[Conference acronym 'XX]{Make sure to enter the correct
  conference title from your rights confirmation email}{June 03--05,
  2018}{Woodstock, NY}
\acmISBN{978-1-4503-XXXX-X/2018/06}

\usepackage{amsmath}
\usepackage{algorithm}
\usepackage{algorithmic}
\usepackage{array}
\usepackage{url}

\usepackage{framed}
\usepackage{multirow}
\usepackage{xspace}
\usepackage{enumitem}
\setlist{leftmargin=*}
\newcommand{\etal}{{\em et al.}\xspace}
\newcommand{\ie}{{\em i.e.,}\xspace}
\newcommand{\eg}{{\em e.g.,}\xspace}
\newcommand{\etc}{{etc.}\xspace}

\usepackage{caption}
\usepackage{subcaption}
\usepackage{graphicx}
\usepackage{booktabs}
\usepackage{makecell}

\newcommand{\BfPara}[1]{{\vspace{1.5mm}\noindent\bf#1.}\xspace}

\usepackage{soul}

\usepackage[most]{tcolorbox}
\tcbset{
  colback=gray!5!white,
  colframe=gray!75!black,
  fonttitle=\bfseries,
  boxrule=0.4pt,
  arc=3pt,
  left=6pt,
  right=6pt,
  top=6pt,
  bottom=6pt,
  enhanced
}

\begin{document}

\title{Multi-Channel Spread-Spectrum Code Watermarking}

\author{Soohyeon Choi}
\email{shchoi@smu.edu.sg}
\orcid{0009-0002-1252-2263}
\affiliation{%
  \institution{Singapore Management University}
  \country{Singapore}
}
\author{Debin Gao}
\email{dbgao@smu.edu.sg}
\orcid{0000-0001-9412-9961}
\affiliation{%
  \institution{Singapore Management University}
  \country{Singapore}
}

\author{Yue Duan}
\email{yueduan@smu.edu.sg}
\orcid{0000-0000-0000-0000}
\affiliation{%
  \institution{Singapore Management University}
  \country{Singapore}
}

\renewcommand{\shortauthors}{Choi \etal}

\begin{abstract}
Attributing code to the large language model that produced it is essential for provenance, licensing, and misuse accountability, yet no deployed watermark meets this need. Generation-time schemes require access to the producing model and cannot be applied to third-party code, while post-hoc schemes work on any code but carry at most 4 bits of payload, far too few to distinguish the many deployed model configurations. 
We present multi-channel spread-spectrum watermarking, the first post-hoc, training-free code watermark with a 24-bit payload and formal robustness guarantees. The scheme encodes bits in variable naming conventions and in eight pairs of semantically equivalent code patterns, and a keyed pseudo-random permutation maps every site to a codeword bit so that each bit receives multiple independent votes. Majority voting absorbs distributed corruption, while an outer Reed-Solomon code recovers the identifier when concentrated channel attacks defeat the vote, yielding provable robustness bounds for formatting, syntactic, and structural attacks.
Across 1,750 Python files from CodeNet and from GPT-4.1 and Llama-4 generations, the watermark achieves 100\% clean-detection accuracy with zero false positives. Under 17 attack types, it recovers the identifier at 97.6\% accuracy under 8 variable renames and 94.1\% under 10\% random per-site corruption, while the strongest post-hoc baseline collapses to 0\% under any single-transform attack. Embedding and detection together take under 200 ms on CPU without training data or GPU.

\end{abstract}

\ccsdesc[500]{Security and privacy~Digital rights management}
\ccsdesc[500]{Security and privacy~Software and application security}
\ccsdesc[300]{Theory of computation~Error-correcting codes}
\ccsdesc[500]{Computing methodologies~Artificial intelligence}

\keywords{Code watermarking, Large language models, Model attribution, Multi-channel watermarking, Spread-spectrum, Reed-Solomon codes, Post-hoc watermarking}

\received{20 February 2007}
\received[revised]{12 March 2009}
\received[accepted]{5 June 2009}

\maketitle

\section{Introduction}
Large language models (LLMs) for code generation have moved rapidly from research demonstration to mainstream developer tool. Tools such as GitHub Copilot~\cite{chen2021evaluating}, GPT-4~\cite{openai2023gpt}, Code Llama~\cite{roziere2023code}, and Qwen-Coder~\cite{hui2024qwen} are now integrated into IDEs, pull-request review systems, and enterprise development pipelines, producing a substantial fraction of the code shipped to public and private repositories. This deployment has created three classes of accountability problems that share a single technical root, the reliable attribution of a given piece of code to the model that produced it.

The first class is security. LLM-generated code has been shown to contain exploitable vulnerabilities at rates comparable to or exceeding human-written code~\cite{pearce2022asleep,perry2023do,sandoval2023lost}, and when such code reaches production, incident response requires attributing the vulnerable artifact to a specific model release before the provider can be notified and affected codebases systematically audited. The second is licensing. Large code models are known to memorize and reproduce training data verbatim under adversarial extraction~\cite{carlini2021extracting}, and when an alleged infringement surfaces, determining which provider is liable requires proof that the suspect artifact originated from their model. The third is academic integrity, where instructors increasingly need not just to flag LLM-generated submissions but to identify which model produced them, since institutional policies often permit some tools while prohibiting others, and public detectors address neither task reliably~\cite{sadasivan2023can}. All three scenarios play out the same way from the provider's perspective. An LLM provider generates code, distributes it to users, and later needs to prove that a suspicious artifact originated from a specific model release. The adversary is the end user or a downstream party, who may reformat the code, rename variables, insert or remove comments, or pass the code through another LLM to launder provenance. Any workable solution must bind a model identifier to the generated code at or shortly after generation, survive these modifications, and avoid false alarms on unrelated code.

\BfPara {Limitations of Existing Techniques} Two families of techniques have been proposed for this goal, and both fall short for third-party attribution where the LLM is unavailable at verification time. \emph{Authorship detection} trains a classifier to decide whether code is LLM-generated~\cite{choi2025attributing,choi2025find}. Its accuracy depends on a statistical gap between human and model code that narrows as models improve, and recent analyses show that such detectors can be defeated by paraphrasing, translation, or light editing~\cite{sadasivan2023can}. Detection is also binary only (whether code is LLM-generated), not which model produced it, so it cannot be used for attribution. \emph{Code watermarking} embeds a signal into the generated code itself. Generation-time schemes such as KGW~\cite{kirchenbauer2023watermark}, SWEET~\cite{lee2024wrote}, STONE~\cite{kim2025marking}, and CODEIP~\cite{guan2024codeip} bias the LLM's token sampling and require access to the decoding process, carrying either zero payload or a payload validated only under narrow attacks. Post-hoc schemes drop the model-access requirement but either carry no payload and collapse under any single-transform attack~\cite{li2025efficient} (see Section~\ref{sec:baseline_comparison}) or require end-to-end neural training for just a 4-bit identifier, supporting only 16 model IDs (SrcMarker~\cite{yang2024srcmarker}, RoSeMary~\cite{zhang2025robust}). 

In short, no existing code-LLM watermarking system simultaneously delivers a payload large enough for the model configurations deployed today, robustness against realistic code transformations, and training-free operation on CPU.

\BfPara {Our Approach} To address this gap, we propose multi-channel spread-spectrum watermarking, the first post-hoc, training-free code watermark with a large (24-bit) payload and formal robustness guarantees for formatting, syntactic, and structural attacks. Our key insight is that pairing two independent families of semantic-preserving code modification with spread-spectrum voting forces the attacker to corrupt both families at once, doubling the modification budget required to break the watermark compared to any single-channel scheme. Bits are spread across two channel types. \textbf{Variable naming conventions} carry 2 bits per variable, encoded in the choice among four naming styles (\texttt{snake\_case}, \texttt{camelCase}, \texttt{PascalCase}, and \texttt{ALL\_CAPS}). \textbf{Structural code patterns} carry 1 bit per site, encoded in eight pairs of semantically equivalent syntactic forms listed in Table~\ref{tab:channels}.

Every available site participates in the watermark through a keyed pseudo-random permutation that maps sites to codeword bit positions~\cite{cox1997secure, collberg2002watermarking}, so each codeword bit receives votes from multiple independent sites. Majority voting across these votes gives a first layer of error correction, absorbing distributed corruption. A Reed-Solomon (RS) outer code~\cite{reed1960polynomial} gives a second layer, recovering the identifier when concentrated channel attacks defeat the vote. The two channel families are independent, since renaming variables does not affect structural patterns and modifying operators does not affect variable names. Breaking the watermark, therefore, requires a larger total modification budget than any attack on one family alone could achieve.





Figure~\ref{fig:listings} illustrates this on a small Python function watermarked with the identifier ``gpt4''\footnote{The short example is for exposition. Its four or so embedding sites are not enough to carry a full 24-bit identifier. Minimum capacity requirements are analyzed in Section~\ref{sec:eval}.}.
The watermarked version differs from the original in three places, each carrying watermark bits.

\begin{figure}[t]
{\setlength{\fboxsep}{1pt}%
\begin{minipage}[t]{0.48\columnwidth}
\begin{lstlisting}[label=lst:original,numbers=none,xleftmargin=0.6em,framexleftmargin=0.6em,basicstyle=\ttfamily\footnotesize,escapeinside={(*}{*)}]
def check_prime((*\colorbox{blue!20}{\strut\texttt{num\_val}}*)):
    if (*\colorbox{red!20}{\strut\texttt{num\_val <= 1}}*):
        return False
    for i in range(2, (*\colorbox{blue!20}{\strut\texttt{num\_val}}*)):
        if (*\colorbox{red!20}{\strut\texttt{num\_val \% i == 0}}*):
            return False
    return True
\end{lstlisting}
\centering\small (a) Original
\end{minipage}\hfill
\begin{minipage}[t]{0.48\columnwidth}
\begin{lstlisting}[label=lst:watermarked,numbers=none,xleftmargin=0.6em,framexleftmargin=0.6em,basicstyle=\ttfamily\footnotesize,escapeinside={(*}{*)}]
def check_prime((*\colorbox{blue!20}{\strut\texttt{NumVal}}*)):
    if (*\colorbox{red!20}{\strut\texttt{1 >= NumVal}}*):
        return False
    for i in range(2, (*\colorbox{blue!20}{\strut\texttt{NumVal}}*)):
        if (*\colorbox{red!20}{\strut\texttt{0 == NumVal \% i}}*):
            return False
    return True
\end{lstlisting}
\centering\small (b) Watermarked (``gpt4'')
\end{minipage}}
\caption{LLM-generated code (a) and its watermarked counterpart (b). Blue marks the variable-naming channel, red marks structural channels.}
\label{fig:listings}
\end{figure}

\begin{enumerate}
  \item \textbf{Variable naming:} \texttt{num\_val} becomes \texttt{NumVal}. The 2-bit variable-naming slot carries the value \texttt{10} (\texttt{PascalCase}).
  \item \textbf{Comparison direction:} \texttt{num\_val <= 1} becomes \texttt{1 >= NumVal}. The 1-bit structural slot carries the value \texttt{1} (flipped from the canonical form).
  \item \textbf{Equality operand order:} \texttt{num\_val \% i == 0} becomes \texttt{0 == NumVal \% i}. The 1-bit structural slot carries the value \texttt{1} (operands swapped from the canonical form).
\end{enumerate}

Each modification is semantic-preserving, and the watermarked code produces identical output for all inputs. The detector, given the secret key, reads choices back and recovers the embedded identifier.

Our evaluation validates this design across 1,750 Python files drawn from CodeNet~\cite{puri2021project} and from GPT-4.1 and Llama-4 generations, under 17 attack types. The system recovers the 24-bit identifier with 100\% no-attack accuracy and zero false positives, matching the theoretical bound of $1/2^{24}$. Robustness reaches 97.6\% under 8 variable renames, 100\% under 32 renames on larger files, and 94.1\% under 10\% random per-site corruption. Embedding completes in 65 to 145 ms and detection in 28 to 50 ms, with no GPU or training data required.

\BfPara{Contributions} The contributions of this paper are summarized as follows:
\begin{enumerate}
  \item A \textbf{multi-channel spread-spectrum embedding} that distributes watermark bits across variable naming conventions and eight structural code patterns, gaining robustness through redundant voting across independent channels.
  \item A \textbf{two-layer error correction} architecture combining spread-spectrum majority voting with RS codes, with formal robustness theorems for formatting, syntactic, and structural attacks, and an empirical capacity-robustness tradeoff in which lower $t$ can outperform higher $t$ on large files.
  \item A \textbf{24-bit payload} with false positive rate $1/2^{24}$ and post-hoc deployment in 65 to 145 ms, requiring no model access, no training data, and no GPU, in contrast to prior post-hoc multi-bit schemes that depend on thousands to millions of training samples.
  \item An \textbf{evaluation on 1,750 files} across seven datasets and 17 attack types, with head-to-head baseline comparisons and a formal proof of the information-theoretic barrier under LLM regeneration.
\end{enumerate}

\section{Background}\label{sec:background}
\subsection{Code Watermarking}
Code watermarking embeds a hidden signal into source code that can later be detected to determine provenance~\cite{collberg2002watermarking}. Unlike natural language watermarking~\cite{kirchenbauer2023watermark}, code watermarking must preserve both semantics (the code produces the same output) and syntactic validity (the code compiles correctly). A single misplaced token can cause a compilation error or a change in the program's behavior.

Existing schemes fall along two axes, namely \emph{when} the watermark is embedded (during generation versus post-hoc) and \emph{what} is embedded (binary presence detection versus multi-bit message). 

Generation-time watermarking modifies the LLM's token sampling process. KGW~\cite{kirchenbauer2023watermark} splits the vocabulary into green and red lists using a hash of the previous token, then adds a bias $\delta$ to green token logits. Detection counts green tokens and computes a $z$-score. SWEET~\cite{lee2024wrote} improves on KGW by applying the bias only to high-entropy tokens, preserving code correctness at low-entropy positions, such as keywords and syntax. STONE~\cite{kim2025marking} further refines this by excluding syntax-critical tokens entirely. All three require access to the LLM's internals and carry zero payload, hence, are only suitable for binary classification (\ie LLM-generated or not).

Post-hoc watermarking techniques alter code after generation, making it applicable to any LLM. ACW~\cite{li2025efficient} applies 45 idempotent code transformations (\eg refactoring, reordering, formatting). Detection re-applies each transformation and checks whether the code changes. If the code is unchanged, the transformation was already applied, and the watermark is considered present. SrcMarker~\cite{yang2024srcmarker} trains a bidirectional gated recurrent unit (BiGRU) encoder with a paired decoder over variable-name transformations and abstract syntax tree (AST) rewrites to embed a 4-bit signature. RoSeMary~\cite{zhang2025robust} trains a CodeT5-based neural encoder~\cite{wang2021codet5} end-to-end to embed a 4-bit signature across variable names and syntactic transformations, with zero-knowledge wrappers around verification.


\subsection{Reed-Solomon Error Correction}
RS codes~\cite{reed1960polynomial} are a family of algebraic error-correcting codes widely used in storage, communication, and QR codes. An RS code with parameters $(n, k, t)$ over $\text{GF}(q)$ encodes $k$ message symbols into $n {=} k {+} 2t$ codeword symbols and corrects any $e$ erasures and $s$ errors satisfying $e {+} 2s {\leq} 2t$. With a Vandermonde parity-check matrix, the code is \emph{maximum-distance separable (MDS)}, and any erasure pattern of size up to $2t$ is provably recoverable.

Our system uses two complementary codes. A random binary linear code~\cite{shannon1948mathematical} serves as the primary code for all main evaluations. A standard RS code over $\text{GF}(16)$ serves as a theoretical baseline with MDS guarantees and is evaluated in the Appendix~\ref{sec:appendix_baselines} on files large enough to carry its longer codeword. Both codes parametrize a capacity-robustness tradeoff through the RS parameter $t$, where a higher $t$ corrects more erasures but requires a longer codeword and therefore larger files.

\BfPara{Compact Code for Short Files} Our primary evaluation uses a systematic binary linear code with $4t$ parity-check bits derived from a fixed SHA-256~\cite{nist2015sha} seed. The parity matrix is random by construction (no Vandermonde structure), so the code is \emph{not} MDS but still decodes with high probability for moderate erasure counts. Specifically, $e {\leq} 2t {-} 1$ erasures recover with probability above 0.99, and $e {=} 2t$ erasures with probability approximately 0.30 (Theorem~\ref{thm:compact}). The codeword fits even the smallest files in $D_{32}$ at $t{=}2$ while keeping the same decoding interface across all datasets.

\BfPara{GF(16) RS Baseline} For theoretical comparison, we evaluate a baseline, which is a standard RS code over $\text{GF}(16)$, whose 4-bit symbols align naturally with our watermark layout, namely 2-bit variable-naming slots and 1-bit structural slots. GF(16) arithmetic operates over $\text{GF}(2)[x]/(x^4{+}x{+}1)$ with primitive element $\alpha {=} x$, and the 24-bit payload encodes as 6 GF(16) symbols. The GF(16) code requires $D_{40}$ or larger files (Table~\ref{tab:rs_params}); Appendix~\ref{sec:appendix_baselines} reports its full per-dataset comparison.

\begin{table}[t]
    \centering
    \begin{tabular}{c cc cc c}
    \toprule
     & \multicolumn{2}{c}{Compact (primary)} & \multicolumn{2}{c}{GF(16) RS (baseline)} &  \\
    \cmidrule(lr){2-3}\cmidrule(lr){4-5}
    $t$ & CW bits & Min.\ bits & CW bits & Min.\ bits & Erasures \\
    \midrule
    2 & 32 & 32 & 40 & 40 & $\leq 4$ \\
    4 & 40 & 40 & 56 & 56 & $\leq 8$ \\
    8 & 56 & 56 & 88 & 88 & $\leq 16$ \\
    \bottomrule
    \end{tabular}
    \caption{Codeword parameters for the compact code (primary) and the GF(16) RS baseline (CW = codeword).}
    \vspace{-5mm}
    \label{tab:rs_params}
\end{table}

\subsection{Spread-Spectrum Watermarking}
Spread-spectrum techniques, originally from wireless communications~\cite{cox1997secure}, distribute a narrowband signal across a wide frequency band, making the signal robust against narrowband interference. We adapt this principle to code watermarking by spreading the codeword across all available sites. Without spreading, each codeword bit depends on one site and a single modification flips it; with spreading, every site contributes to multiple bits.

A keyed pseudo-random permutation assigns each site to a codeword bit position. The key generation is implementation-specific; we use a permutation seeded by a point sequence on the secp256k1 elliptic curve, detailed in Section~\ref{sec:design}. With $N$ sites and $B$ codeword bits, each bit receives approximately $v {=} \lfloor N/B \rfloor$ votes. For example, $D_{56}$ at $t{=}2$ yields $v {\approx} 1.31$ votes per bit ($N {\approx} 42$, $B {=} 32$); $D_{152}$ yields $v {\approx} 4.29$ from $N {\approx} 137$.

This design provides two layers of error correction that operate independently.
\begin{enumerate}
  \item \textbf{Majority voting} forms the first layer. If a codeword bit receives $v$ votes and fewer than $\lfloor v/2 \rfloor$ are corrupted, the majority vote is correct, and no RS correction is needed for this position.
  \item \textbf{RS decoding} forms the second layer. For positions where the majority vote fails, RS correction handles up to $2t$ such positions across the entire codeword.
\end{enumerate}

The two layers address complementary threats. The spread-spectrum layer absorbs randomly distributed corruption from individual site modifications, while the RS layer handles concentrated corruption, such as an attacker targeting an entire channel.

\section{Threat Model and Design Goals}\label{sec:threat}
\subsection{Deployment Model}
We consider a \emph{provider-as-verifier} deployment model, standard across all keyed watermarking systems~\cite{kirchenbauer2023watermark,lee2024wrote,kim2025marking,qu2025provably}:
\begin{enumerate}
  \item An LLM provider generates code and embeds a watermark containing the model identifier (\eg ``gpt4'') using a secret key.
  \item The watermarked code is distributed to users.
  \item A user may modify the code (\eg renaming variables, reformatting, refactoring).
  \item A verifier (\eg academic institution, employer, court) submits the code to a detection service.
  \item The service provider tests it against all registered LLM provider keys and returns the result (``Generated by GPT-4'' or ``no watermark'').
\end{enumerate}

Detection requires the secret key $K$, preventing users from running detection themselves. With detection taking $\approx 40$ms per key and a small number of major LLM providers, exhaustive checking completes in under 1 second, analogous to how YouTube's Content ID checks uploaded videos against all registered content~\cite{soha2016monetizing}.

\subsection{Attacker Model}
We assume the adversary does not possess the secret key $K$, which the provider-as-verifier deployment model keeps private. The adversary, therefore, knows the watermarking scheme and channel types but cannot determine which sites map to which codeword bit positions.

We define a four-level attack taxonomy based on the type of code transformation. Our system provides formal guarantees for Levels 1--3. Level 4 is an information-theoretic barrier for \emph{all} post-hoc syntactic watermarks (Section~\ref{sec:llm_barrier}).

\begin{itemize}
  \item \textbf{Level 1: Formatting.} Whitespace, comments, blank lines, and import ordering. These changes do not affect the AST and therefore corrupt zero watermark sites (Theorem~\ref{thm:level1}).
  \item \textbf{Level 2: Syntactic.} Variable renaming, dead code insertion, and statement reordering, which are classical code-transformation operations~\cite{collberg2002watermarking}. Each variable rename corrupts at most 2 bits via the naming convention. Dead code insertion may shift AST-order site indices.
  \item \textbf{Level 3: Structural.} Operator direction swapping (\texttt{a<b} $\leftrightarrow$ \texttt{b>a}), operand reordering (\texttt{a+b} $\leftrightarrow$ \texttt{b+a}), and expression-form changes (ternary $\leftrightarrow$ if/else, list comprehension $\leftrightarrow$ for-loop). Each channel attack corrupts all sites of that channel.
  \item \textbf{Level 4: Algorithmic.} Complete code regeneration by another LLM. All syntactic features are independently regenerated from the program's semantic content.
\end{itemize}

\subsection{The Detection Dilemma}\label{sec:oracle}
A fundamental tension exists in all watermarking systems between verifiability and security. Any detection mechanism that returns a binary answer can be queried as an oracle for iterative removal: an adversary modifies the code incrementally, queries the detector after each modification, and continues until detection fails. For code files with 10 to 30 modifiable sites, this requires only $O(N)$ queries.

The dilemma is universal across prior watermarking work and remains unresolved. KGW~\cite{kirchenbauer2023watermark} acknowledges the risk and proposes access monitoring without formalization. Qu~\etal~\cite{qu2025provably} prove security under the Random Oracle Model~\cite{bellare1993random} against blind, non-adaptive attacks only. RoSeMary~\cite{zhang2025robust} wraps verification in zero-knowledge proofs (ZKP), but the binary outcome is itself an oracle signal regardless of whether the key is hidden. Other schemes~\cite{lee2024wrote,kim2025marking,guan2024codeip,li2025efficient} do not address the risk.
We adopt the provider-as-verifier model to prevent user-side oracle queries by default. Practical mitigations via rate limiting and query deduplication are discussed in Section~\ref{sec:discussion}.

\subsection{Design Goals}
\label{sec:goals}

The threat model above motivates five concrete goals that the system must satisfy together in order to be practical.

\BfPara{G1: Multi-bit Attribution} The watermark must carry enough payload to distinguish among the model identifiers deployed in practice today. We target 24 bits, sufficient for $2^{24} {\approx} 1.7 {\times} 10^7$ provider-model pairs, well beyond the 4-bit (16 IDs) ceiling of prior post-hoc neural schemes~\cite{yang2024srcmarker,zhang2025robust}.

\BfPara{G2: Robustness Under Attack} Detection must survive every Level 1 attack deterministically (Theorem~\ref{thm:level1}) and tolerate Level 2 and Level 3 attacks within the recovery bounds of the spread-spectrum voting and RS correction layers (Theorems~\ref{thm:rename} and ~\ref{thm:struct}).

\BfPara{G3: Semantic Preservation} Every embedding transformation must preserve program behavior on all inputs, and the watermarked code must still compile without modification.

\BfPara{G4: No Training or Model Access} The provider must embed and verify without model access, training data, or GPU resources. This rules out neural encoders trained end-to-end and any scheme that intercepts the LLM's decoding process.

\BfPara{G5: Low False-Positive Rate} The detector must reject unrelated code with vanishingly small probability. We target a provable bound of $1/2^{24}$, matching the payload's information-theoretic floor.


\section{System Design} \label{sec:design}

\subsection{Overview}
\begin{figure}[t]
  \centering
  \includegraphics[width=\columnwidth]{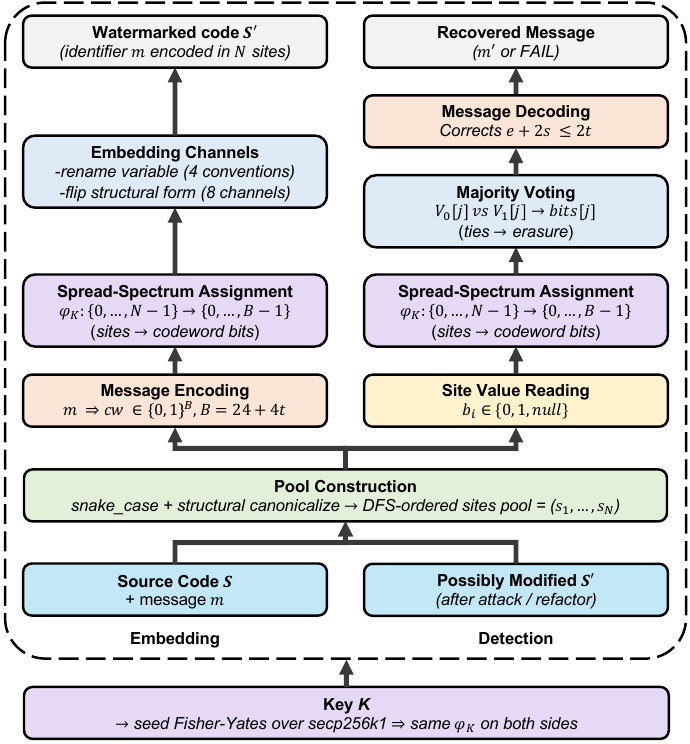}
  \caption{System overview of the embedding (left) and detection (right) pipelines.}
  \label{fig:overview}
\end{figure}
Figure~\ref{fig:overview} illustrates the two pipelines, \emph{embedding} and \emph{detection}. Both share identical normalization and pool construction steps, ensuring that the same source code always produces the same site ordering, which is critical for correct detection.

\subsection{Message Encoding}
The watermark payload is a model identifier string, encoded with a 6-bit compact alphabet of 64 characters, namely \texttt{a--z} (0--25), \texttt{A--Z} (26--51), \texttt{0--9} (52--61), \texttt{\_} (62), and null padding (63). This alphabet covers every plausible model identifier (\eg ``gpt4'', ``dsk2''), since deployed model names consist of alphanumeric characters and underscores. A 4-character identifier therefore produces up to $4 {\times} 6 {=} 24$ message bits. The compact code appends $4t$ parity bits to yield a $(24{+}4t)$-bit codeword, whereas the $\text{GF}(16)$ baseline encodes the same 24 bits as 6 four-bit symbols (Section~\ref{sec:background}). For example, ``gpt4'' encodes as $g{=}6$, $p{=}15$, $t{=}19$, $4{=}56$, giving the bit string \texttt{000110 001111 010011 111000}. Compared with 8-bit ASCII, this encoding cuts the minimum capacity requirement by 25\%, from 32 bits to 24 bits, enabling watermarking of shorter code files without sacrificing coverage of any deployed model name.

\subsection{Embedding Channels}
The watermark spreads bits across two complementary channel families, a variable naming channel and eight structural channels.

\BfPara{Variable Naming} Each multi-word variable carries 2 bits via the choice among four naming conventions, namely \texttt{snake\_case} (bits \texttt{00}), \texttt{camelCase} (bits \texttt{01}), \texttt{PascalCase} (bits \texttt{10}), and \texttt{ALL\_CAPS} (bits \texttt{11}). The pool construction step, described below, normalizes all variables to \texttt{snake\_case} and extends single-word identifiers to multi-word forms so they can carry a detectable convention. The embedding algorithm then renames each variable to the convention determined by its codeword bit assignment.

\begin{table}[t]
    \centering
    \begin{tabular}{@{}clll@{}}
    \toprule
    \# & Channel & Canonical (bit=0) & Alternative (bit=1) \\
    \midrule
    1 & aug\_assign    & \texttt{x = x + 1}  & \texttt{x += 1} \\
    2 & compare\_dir   & \texttt{a < b}      & \texttt{b > a} \\
    3 & equality\_dir  & \texttt{y == x}     & \texttt{x == y} \\
    4 & list\_comp     & for + append        & comprehension \\
    5 & ternary        & if/else block       & \texttt{x if c else y} \\
    6 & operand\_order & \texttt{a + b}      & \texttt{b + a} \\
    7 & empty\_coll    & \texttt{[]}         & \texttt{list()} \\
    8 & none\_check    & \texttt{== None}    & \texttt{is None} \\
    \bottomrule
    \end{tabular}
    \caption{The eight structural channels for Python.}
    \label{tab:channels}
    \vspace{-4mm}
\end{table}

\BfPara{Structural Channels} Each of the eight structural channels (listed in Table~\ref{tab:channels}) exploits a pair of semantically equivalent syntactic forms, and the watermark encodes one bit per site by choosing between the canonical form (bit 0) and the alternative form (bit 1). All transformations are semantic-preserving, since both forms produce identical program output. The channels are also independent, so attacking one channel (\eg normalizing all comparisons to the canonical `\texttt{<}' form) does not affect other channels (\eg \texttt{augmented assignment} or \texttt{operand order}). Channels vary in prevalence, with \texttt{operand\_order} contributing the most sites per file (8.1 on average in $D_{56}$, 37\% of structural capacity) and \texttt{ternary} and \texttt{none\_check} the fewest (0.2 and 0.3 sites respectively).

\begin{algorithm}[h]
    \caption{Watermark Embedding}
    \label{alg:embed}
    \begin{algorithmic}[1]
    \REQUIRE Source code $S$, message $m$, key $K$, RS parameter $t$
    \ENSURE Watermarked code $S'$
    \STATE $\text{canon} \leftarrow \text{Normalize}(S)$ \COMMENT{snake\_case + structural}
    \STATE $\text{pool} \leftarrow \text{ConstructPool}(\text{canon})$
    \STATE $N \leftarrow |\text{pool}|$
    \STATE $\text{cw} \leftarrow \text{RS\_Encode}(m, t)$ \COMMENT{$B = 24 + 4t$ bits}
    \STATE $\phi \leftarrow \text{SpreadSpectrum}(K, N, B)$ \COMMENT{keyed permutation}
    \FOR{$i = 0$ to $N-1$}
      \STATE Read target bits for site $i$ from $\text{cw}$ via $\phi$ \COMMENT{variables use 2 consecutive $\phi$ slots}
      \STATE Transform $\text{pool}[i]$ to encode those bits
    \ENDFOR
    \RETURN watermarked source $S'$
    \end{algorithmic}
\end{algorithm}

\subsection{Pool Construction}\label{sec:pool-construction}
The site pool determines which embeddable sites carry watermark bits. Pool construction must be \emph{deterministic}, such that the same source code and the same key always yield the same pool in the same order. The process proceeds in five steps.

\begin{enumerate}
  \item \textbf{Normalize variable names} to \texttt{snake\_case} (canonical form), so that variables are identified consistently regardless of the original naming style.
  \item \textbf{Extend single-word variables} with context-inferred suffixes (\eg \texttt{n} {$\rightarrow$} \texttt{n\_count}, \texttt{s} {$\rightarrow$} \texttt{s\_str}) so they become multi-word and can carry a detectable naming convention; single-character variables such as \texttt{`i'} cannot, since \texttt{snake\_case} and \texttt{camelCase} are indistinguishable for single words. AST-validity tests (Section~\ref{sec:eval}) confirm that the renamed code remains functionally correct.
  \item \textbf{Normalize structural channels} to canonical forms. All comparisons are normalized to \texttt{<}/\texttt{<=} (flipping \texttt{>}/\texttt{>=}), commutative operands are sorted by type, and equality operands are ordered by complexity.
  \item \textbf{Extract embeddable sites} in deterministic depth-first search order, naming-convention variables first, then structural sites.
  \item \textbf{Construct the ordered pool} as a list of (site, bit-width) entries, where variables contribute 2 bits each and structural sites contribute 1 bit each.
\end{enumerate}

\subsection{Spread-Spectrum Bit Assignment}
With $N$ sites and a $B$-bit codeword, the key $K$ yields an assignment $\phi_K {:} \{0,\ldots,N{-}1\} {\to} \{0,\ldots,B{-}1\}$ via a Fisher-Yates shuffle~\cite{knuth1997art} seeded by a point sequence on the secp256k1 elliptic curve~\cite{miller1986use, secg2010sec2}. The $x$-coordinates of the scalar-multiplied sequence $K{\cdot}G, 2K{\cdot}G, \ldots$ on secp256k1 drive the Fisher-Yates swap indices.

The assignment operates in full permutation cycles, where each cycle of $B$ assignments covers all codeword bit positions exactly once. With $N$ sites, the number of full cycles is $\lfloor N/B \rfloor$, and each codeword bit receives between $\lfloor N/B \rfloor$ and $\lceil N/B \rceil$ votes. This uniform distribution is critical because it ensures no single codeword bit position is under-represented.

\subsection{Embedding and Detection Algorithms}
We now present both pipelines as pseudocode. The embedding pipeline (Algorithm~\ref{alg:embed}) takes the source code, message, key, and RS parameter $t$, and produces watermarked code with the identifier encoded across $N$ sites. Normalization and pool construction establish a deterministic site ordering (Section~\ref{sec:pool-construction}), the RS encoder produces a codeword of $B {=} 24 {+} 4t$ bits, and the spread-spectrum assignment $\phi_K$ maps each site to a codeword bit position. Each site is then transformed to encode its target bit assignment.

The detection pipeline (Algorithm~\ref{alg:detect}) reverses this process. Given possibly modified code, detection reconstructs the same pool using identical normalization, which is critical for correctness. Each site contributes a vote to its assigned codeword bit position. A position with a strict majority of $V_0$ or $V_1$ yields the corresponding bit; ties produce erasures that the RS decoder corrects. If decoding succeeds, the recovered message is returned; otherwise, it outputs FAIL.

\begin{algorithm}[t]
    \caption{Watermark Detection}
    \label{alg:detect}
    \begin{algorithmic}[1]
    \REQUIRE Possibly modified code $S'$, key $K$, RS parameter $t$
    \ENSURE Recovered message $m'$ or FAIL
    \STATE $\text{canon} \leftarrow \text{Normalize}(S')$
    \STATE $\text{pool} \leftarrow \text{ConstructPool}(\text{canon})$
    \STATE $N \leftarrow |\text{pool}|$, $B \leftarrow 24 + 4t$
    \STATE $\phi \leftarrow \text{SpreadSpectrum}(K, N, B)$
    \STATE Initialize vote counters $V_0[0..B{-}1] \leftarrow 0$, $V_1[0..B{-}1] \leftarrow 0$
    \FOR{$i = 0$ to $N-1$}
      \STATE $b \leftarrow \text{ReadSiteValue}(\text{pool}[i], S')$
      \IF{$b \neq \text{null}$}
        \STATE $V_b[\phi(i)] \leftarrow V_b[\phi(i)] + 1$
      \ENDIF
    \ENDFOR
    \FOR{$j = 0$ to $B-1$}
      \IF{$V_0[j] > V_1[j]$}
        \STATE $\text{bits}[j] \leftarrow 0$
      \ELSIF{$V_1[j] > V_0[j]$}
        \STATE $\text{bits}[j] \leftarrow 1$
      \ELSE
        \STATE $\text{bits}[j] \leftarrow \text{ERASURE}$
      \ENDIF
    \ENDFOR
    \STATE $m' \leftarrow \text{RS\_Decode}(\text{bits}, t)$
    \RETURN $m'$ if decoding succeeds, else FAIL
    \end{algorithmic}
\end{algorithm}

\subsection{Concrete Example}
We walk through the mechanism on a typical $D_{56}$ file. Such a file provides about 21 multi-word variables and 22 structural sites. Variables contribute 2 votes each (one per naming-convention bit, occupying two consecutive $\phi_K$ slots as in Algorithm~\ref{alg:embed}) and structural sites contribute 1 vote each, giving $N {=} 21 {\times} 2 {+} 22 {=} 64$ total votes. With $t{=}2$, the codeword has $B {=} 32$ bits, so the spread-spectrum assignment $\phi_K$ averages $v {=} N/B {=} 2.0$ votes per bit position. 
The per-site $V$/bit column reported in Appendix~\ref{sec:appendix_extra} (Table~\ref{tab:crossdataset}) counts one vote per site rather than one per bit and is therefore roughly half of $v$ as used.

Under embedding, $\phi_K$ maps each of the 64 vote slots to one codeword bit position, and the corresponding site is transformed to match its assigned bit. Under detection, the same 64 sites vote on the same 32 positions. An attack that flips 3 structural sites corrupts 3 of the 64 vote slots. The spread-spectrum assignment distributes these wrong votes pseudo-randomly, so each lands at a distinct bit position with high probability. Each affected position has one correct vote and one wrong vote, producing a tie that the detector flags as an erasure. The RS decoder corrects up to $2t{=}4$ erasures, so the three erasures fall within the correction bound, and the identifier is recovered exactly.

Figure~\ref{fig:voting} illustrates the same mechanism on a scaled toy with $N {=} 21$ sites and $B {=} 7$ codeword bits for visibility, where each bit receives $v {=} 3$ votes. Sites 0--9 are variable slots (blue), sites 10--20 are structural sites (orange), and sites 11, 15, and 18 are attacker-flipped (red). Tallies below each codeword bit are written $V_1 {\cdot} 1 {\mid} V_0 {\cdot} 0$, showing the number of votes for each value. With three votes per bit, each flipped vote is outvoted 2-to-1 at its assigned codeword bit, so majority voting recovers every bit directly ($e {=} s {=} 0$), and RS never needs to engage. On production files with $v {\approx} 2$, the detector instead sees ties that become erasures, which RS resolves as described.

\begin{figure}[t]
  \centering
  \includegraphics[width=\columnwidth]{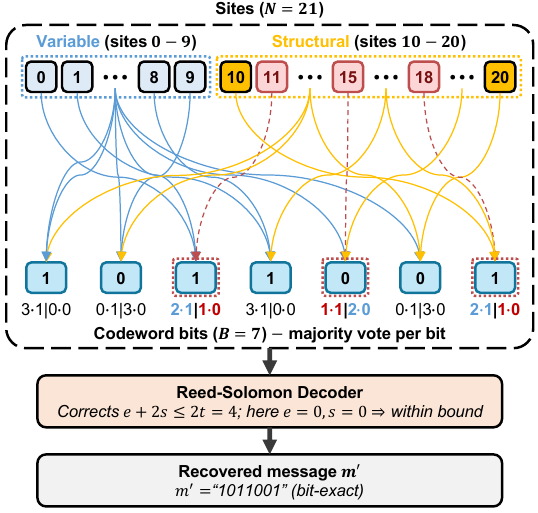}
  \caption{Spread-spectrum voting example. Three attacker-flipped sites are each outvoted 2-to-1 by correct votes, and RS decoding recovers the message without correction.}
  \label{fig:voting}
  \vspace{-2mm}
\end{figure}

\section{Formal Analysis}\label{sec:analysis}
We provide robustness guarantees for Levels 1--3 of our attack taxonomy (Section~\ref{sec:threat}). Our main evaluation uses the compact binary code, whose recovery guarantee is given by Theorem~\ref{thm:compact}. The $\text{GF}(16)$ RS baseline, reported in the Appendix~\ref{sec:appendix_baselines}, satisfies the stronger MDS guarantee of Theorem~\ref{thm:rs}.

\begin{theorem}[GF(16) RS: MDS Guarantee]\label{thm:rs}
For the $\text{GF}(16)$ RS encoding, any received codeword with $e$ symbol erasures and $s$ symbol errors satisfying $e {+} 2s {\leq} 2t$ is decoded to the original message exactly.
\end{theorem}

\begin{proof}
The code is an RS code over $\text{GF}(2^4)$ with $k{=}6$ message symbols and $n {=} k{+}2t$ codeword symbols. The parity-check matrix is $H_{ji} {=} \alpha^{i(j+1)}$ for $i {\in} [0,n)$, $j {\in} [0,2t)$, where $\alpha$ is a primitive element. By the Vandermonde property, every $r {\times} r$ submatrix of $H$ for $r {\leq} 2t$ is invertible over $\text{GF}(16)$, establishing the MDS property with minimum distance $d {=} 2t{+}1$. The combined bound $e {+} 2s {\leq} 2t$ is the standard unique-decoding condition for MDS codes, achieved by Gaussian elimination over $\text{GF}(16)$.
\end{proof}

\begin{theorem}[Compact Code: Probabilistic Guarantee]\label{thm:compact}
For the compact binary code, the decoder recovers the message when the $4t {\times} e$ submatrix of $H$ indexed by the $e$ erased bit positions has full column rank. For a uniformly random $H$ and erasure pattern of weight $e$, this probability is $\prod_{i=1}^{4t{-}e{+}1}(1{-}2^{-i})$. Empirically, the decoder recovers ${>}99\%$ of cases for $e {\leq} 2t{-}1$ and ${\approx}30\%$ for $e {=} 2t$.
\end{theorem}

\begin{proof}[Proof sketch]
The decoder solves the linear system $H_E\, x_E {=} b$ over $\text{GF}(2)$, where $H_E$ is the $4t {\times} e$ submatrix of $H$ whose columns correspond to the erased positions, $x_E$ is the vector of unknown erased-bit values, and $b$ is the syndrome from the known bits. A unique solution exists iff $H_E$ has full column rank. For a uniformly random parity matrix $H$, this probability follows $\prod_{i=1}^{4t{-}e{+}1}(1{-}2^{-i})$ as stated, which we verify empirically (Appendix~\ref{sec:appendix_extra}, Table~\ref{tab:crossdataset}, $D_{32}$ column) at ${>}0.99$ for $e{\leq}3$ and ${\approx}0.30$ for $e {=} 2t {=} 4$.
\end{proof}

\begin{theorem}[Level 1 Invariance]\label{thm:level1}
Formatting changes (whitespace, comments, blank lines, import ordering) affect zero watermark sites. The watermark is perfectly preserved.
\end{theorem}

\begin{proof}
All watermark sites are extracted from the AST, which is invariant to whitespace, comments, blank lines, and import ordering. The normalization step (Section~\ref{sec:design}) produces identical canonical forms regardless of formatting, so the pool and all site values are preserved.
\end{proof}

\begin{theorem}[Rename Robustness]\label{thm:rename}
Renaming $r {\leq} t$ variables corrupts at most $2r$ bit positions in the codeword. For the $\text{GF}(16)$ code, the RS decoder corrects all such corruptions exactly (Theorem~\ref{thm:rs}). For the compact code, correction succeeds with the probabilistic guarantee of Theorem~\ref{thm:compact}.
\end{theorem}

\begin{proof}
Each variable encodes 2 bits via its naming convention, so renaming $r$ variables corrupts at most $2r$ bit positions in the codeword. For the $\text{GF}(16)$ code, these $2r$ bits map to at most $2r$ symbol positions in the worst case; treating them as erasures gives $e {=} 2r {\leq} 2t$ when $r {\leq} t$, so Theorem~\ref{thm:rs} guarantees exact correction. For the compact code, the same $2r$ positions are treated as erasures and Theorem~\ref{thm:compact} gives the recovery probability.
The bound is tight, since renaming $t{+}1$ variables can produce $2(t{+}1)$ corrupted positions, exceeding the $2t$ budget. Empirically, $t{=}8$ with 8 renames recovers 97.6\% on $D_{56}$ (Table~\ref{tab:main_results}).
\end{proof}

\begin{theorem}[Structural Channel Robustness]\label{thm:struct}
Let $N$ be the total number of sites and $B$ the codeword length in bits. Consider an attacker who corrupts $c$ structural sites chosen without knowledge of $\phi_K$. Over the uniform spread-spectrum assignment, the expected fraction of corrupted votes at each codeword bit position is $c/N$. When $c {<} N/2$, each position is expected to retain a correct majority, and the RS outer code absorbs residual errors up to its $e {+} 2s {\leq} 2t$ bound (Theorem~\ref{thm:rs}).
\end{theorem}

\begin{proof}
The assignment $\phi_K$ distributes sites uniformly across $B$ bit positions in full permutation cycles, so each position receives $v {\approx} N{/}B$ votes. By symmetry, the expected number of corrupted votes at position $j$ is $\mathbb{E}[c_j] {=} cv {/} N$, so the majority condition $c_j {<} v/2$ holds in expectation when $c {<} N{/}2$.
Variance around this expectation can still tip individual positions past the majority threshold, particularly when $c$ approaches $N{/}2$ or when corrupted sites concentrate in a single channel. The RS outer code absorbs such residual errors up to the $2t$ bound. For $t{=}8$ on $D_{56}$, $N {\approx} 66$ and attacking the hardest single structural channel corrupts $c {\approx} 8.1$ votes (12\% of total), well below $N{/}2$. Empirically, this recovers 29.9\% of files, reflecting the concentrated-channel case where variance overwhelms the majority at several positions before RS correction engages; weaker channel attacks yield higher recovery rates (Section~\ref{sec:eval}).
\end{proof}

\begin{theorem}[Attack Budget]\label{thm:keyunaware}
Let $\phi_K$ be a keyed permutation mapping $N$ sites to $B$ codeword bit positions in full permutation cycles, with $v {=} {\lfloor} N{/}B {\rfloor}$ votes per bit. For an attacker who corrupts $c$ sites without knowledge of $K$, breaking RS decoding requires flipping the majority vote at strictly more than $2t$ codeword bit positions (by Theorem~\ref{thm:rs}). Each such flip requires at least ${\lceil} v{/}2 {\rceil}$ corrupted votes to land on that position. In expectation over the uniform permutation, the attacker therefore needs
$c {\geq} (2t{+}1) {\cdot} {\lceil} v{/}2 {\rceil} {=} \Omega\left(\frac{tN}{B}\right)$
corrupted sites to break detection.
\end{theorem}

\begin{proof}
RS decoding fails only when more than $2t$ symbol positions are corrupted. To corrupt a codeword bit, the attacker must overturn its majority vote, which requires at least ${\lceil} v{/}2 {\rceil}$ of the $v$ votes at that position to be flipped. Since $\phi_K$ is a uniformly random permutation and the attacker does not know which sites map to a given bit, the expected number of sites the attacker must corrupt to land ${\lceil} v{/}2 {\rceil}$ on any specific bit is $(\lceil v/2 \rceil) {\cdot} (N{/}v) {=} N/2$ per target bit. Targeting $2t{+}1$ bits simultaneously requires $c {\geq} (2t{+}1) {\cdot} {\lceil} v{/}2 {\rceil}$ corruptions by linearity of the vote counts, which simplifies to $\Omega(tN{/}B)$ for $B {\leq} N$.
\end{proof}

\begin{theorem}[Level 4 Impossibility]\label{thm:impossible}
Let $(\text{embed}, \text{detect})$ be a post-hoc syntactic watermark scheme whose output is obtained by applying only syntactic transformations (variable names, operator forms, or expression structure) that do not change the input-output behavior of the program. Let $T$ be any regeneration function whose output distribution on input $P$ depends only on the input-output behavior of $P$. Then for any program $S$, message $m$, and key $K$,
\[\Pr[\text{detect}(T(\text{embed}(S, m, K)), K) = m] \approx 1/|M|,\]
where $|M|$ is the message space size.
\end{theorem}

\begin{proof}
By the first premise, $\text{embed}(S, m, K)$ and $S$ are behaviorally indistinguishable. By the second premise, $T$'s output distribution depends only on behavior, so $T(\text{embed}(S, m, K))$ and $T(S)$ are identically distributed. The random variable $T(\text{embed}(S, m, K))$ therefore carries no information about $m$. Detection succeeds only when $T$'s syntactic choices happen to match the encoded message, which occurs with probability $1/|M| 
{=} 1/2^{24} {\approx} 6 {\times} 10^{-8}$.
\end{proof}

\section{Evaluation}\label{sec:eval}
\begin{table*}[t]
    \centering
    \begin{tabular}{lcrrrrrr}
    \toprule
    Dataset & Capacity range & Total files & Eval $n$ & Avg vars & Avg struct & Avg total & Avg LOC \\
    \midrule
    $D_{32}$  & 32--39 bits    & 15,201 & 250 & 11.9 & 10.9 &  34.6 &   41.5 \\
    $D_{40}$  & 40--55 bits    & 12,308 & 250 & 15.6 & 14.4 &  45.6 &   59.8 \\
    $D_{56}$  & 56--87 bits    &  5,787 & 250 & 21.9 & 22.1 &  66.0 &   94.3 \\
    $D_{88}$  & 88--151 bits   &  1,419 & 250 & 33.8 & 34.2 & 101.7 &  152.8 \\
    $D_{152}$ & ${\geq}$152 bits &    300 & 250 & 70.6 & 74.6 & 201.8 &  463.1 \\
    \midrule
    gpt41  & LLM-generated & 268 & 250 & 32.1 & 26.5 &  92.1 &  108.8 \\
    llama4 & LLM-generated & 268 & 250 & 18.1 & 15.2 &  51.7 &   58.7 \\
    \bottomrule
    \end{tabular}
    \caption{Dataset summary. Each $D_k$ covers a disjoint capacity range. All evaluations use the compact random binary code.}
    \label{tab:datasets}
    \vspace{-2mm}
\end{table*}

\subsection{Experimental Setup}
\BfPara{Dataset} We draw human-written code from Project CodeNet~\cite{puri2021project}, a large-scale dataset of competitive programming solutions. We partition the Python subset into five disjoint capacity buckets, each covering a non-overlapping range of embedding capacity, so that cross-dataset comparisons isolate capacity effects rather than confounding with file-size overlap, and sample 250 files per bucket. To evaluate on real LLM-generated code, we extract 250 problem statements from CodeNet and regenerate solutions with GPT-4.1 and Llama-4-Maverick, yielding two additional datasets of 250 files each. Table~\ref{tab:datasets} summarizes all seven evaluation datasets.

\BfPara{Attacks} Each attack family probes a distinct design surface. Rename attacks isolate the variable-naming channel, per-channel structural attacks exercise each structural channel individually, and graduated per-site attacks stress the spread-spectrum voting layer directly. Beyond a clean no-attack baseline, we evaluate 17 attack types organized by our taxonomy.

\begin{itemize}
  \item \emph{Rename $r$}: rename $r {\in} \{2, 4, 8, 16, 32\}$ variables to random names (Level 2).
  \item \emph{Per-channel structural}: attack all sites of one channel, 8 individual attacks (Level 3).
  \item \emph{Graduated per-site}: randomly flip $x\%$ of individual sites with $x {\in} \{10, 20, 30, 40\}$, measuring precise corruption tolerance.
\end{itemize}

\BfPara{Metrics} Detection accuracy is the fraction of files where the correct model identifier is recovered. False positive rate (FPR) is the fraction of non-watermarked files that trigger a false detection.

\subsection{Correctness, Fidelity, and False Positive Rate}
Table~\ref{tab:fpr} summarizes baseline sanity checks. The system achieves perfect no-attack accuracy across all 1,750 watermarked files, and all watermarked files parse correctly (100\% AST validity). This is by construction, since our transformations are semantic-preserving syntactic rewrites and Python's \texttt{ast.unparse} always produces valid code. The false positive rate of 0/1{,}750 on both unwatermarked code (correct key) and wrong-key detection is consistent with the theoretical bound of $1/2^{24}$.

\begin{table}[t]
    \centering
    \begin{tabular}{llr}
    \toprule
    Test & $n$ & Result \\
    \midrule
    No-attack accuracy & 1,750 & \textbf{100\%} \\
    AST validity (watermarked code parses) & 1,750 & \textbf{100\%} \\
    \midrule
    FPR: unwatermarked code, correct key & 1,750 & \textbf{0.0\%} \\
    FPR: watermarked code, wrong key & 1,750 & \textbf{0.0\%}  \\
    \bottomrule
    \end{tabular}
    \caption{Correctness, fidelity, and false positive rate over 1,750 files (all 7 datasets, $n{=}250$ each). Theoretical FPR bound $1/2^{24} {\approx} 0.000006\%$. Empirical results are consistent with theory.}
    \label{tab:fpr}
    \vspace{-5mm}
\end{table}

\begin{table}[t]
    \centering
    \begin{tabular}{lrrrrr}
    \toprule
     & \multicolumn{3}{c}{RS parameter scan} & \multicolumn{2}{c}{Large-file scaling} \\
    \cmidrule(lr){2-4}\cmidrule(lr){5-6}
    Attack & $t{=}2$ & $t{=}4$ & $t{=}8$ & $t{=}2$ & $t{=}2$ \\
     & ($D_{32}$) & ($D_{40}$) & ($D_{56}$) & ($D_{88}$) & ($D_{152}$) \\
    \midrule
    No attack & 100.0 & 100.0 & 100.0 & 100.0 & 100.0 \\
    Rename 2  &  98.4 &  99.6 &  98.0 & \textbf{100.0} & \textbf{100.0} \\
    Rename 4  &  98.4 &  99.6 &  98.0 & \textbf{100.0} & \textbf{100.0} \\
    Rename 8  &   1.6 &  99.6 &  97.6 & \textbf{100.0} &  99.6 \\
    Rename 16 &  0.0$^\dagger$ & 0.0$^\dagger$ & 16.2 &  99.2$^\dagger$ & \textbf{100.0}$^\dagger$ \\
    Rename 32 & -- & -- & 0.0$^\dagger$ & 59.7$^\dagger$ & \textbf{100.0}$^\dagger$ \\
    \midrule
    aug\_assign    & 78.8 & 63.2 & 70.9 & \textbf{98.0} &  \textbf{98.0} \\
    compare\_dir   & 76.0 & 62.0 & 59.8 & 96.0 &  \textbf{98.8} \\
    equality\_dir  & 58.0 & 52.4 & 49.8 & 90.0 &  \textbf{96.0} \\
    list\_comp     & 82.0 & 82.4 & 76.5 & 99.6 & \textbf{100.0} \\
    ternary        & 97.6 & 96.0 & 95.6 &\textbf{100.0} & \textbf{100.0} \\
    operand\_order & 41.6 & 36.8 & 29.9 & 80.0 &  \textbf{81.6} \\
    empty\_coll    & 83.6 & 79.2 & 74.5 & 99.6 & \textbf{100.0} \\
    none\_check    &\textbf{100.0} & 98.8 & 97.2 & 99.6 & \textbf{100.0} \\
    \bottomrule
    \end{tabular}
    \caption{Detection accuracy (\%) under various attacks ($n{=}250$ per column). First three columns vary $t$; last two fix $t{=}2$ on larger files. $\dagger$: restricted to eligible files, meaning files with enough distinct variables to rename; \texttt{--}: no eligible files in this column.}
    \vspace{-0.2in}
    \label{tab:main_results}
    \vspace{-2mm}
\end{table}

\subsection{Robustness Under Attacks}\label{sec:robustness}
We first isolate rename and structural attacks on our own system (Table~\ref{tab:main_results}), then compare against post-hoc baselines in Section~\ref{sec:baseline_comparison} and against the generation-time paradigm in Section~\ref{sec:stone_comparison}. Each $t$ is deliberately evaluated on the smallest dataset that can carry its codeword, showing robustness under the hardest capacity conditions. The last two columns demonstrate that on larger files, even $t{=}2$ provides strong robustness through majority voting alone, often matching or exceeding higher $t$ on structural attacks.

Higher $t$ improves rename robustness by increasing the RS correction budget. The improvement is most dramatic for rename attacks: \texttt{rename\_8} improves from 1.6\% ($t{=}2$, $D_{32}$) to 97.6\% ($t{=}8$, $D_{56}$), confirming Theorem~\ref{thm:rename}.
For structural attacks, however, the pattern reverses: higher $t$ can \emph{reduce} structural robustness. For example, \texttt{operand\_order} drops from 41.6\% ($t{=}2$, $D_{32}$) to 29.9\% ($t{=}8$, $D_{56}$). The reason is that a higher $t$ requires a longer codeword ($24{+}4t$ bits), which dilutes votes per bit when the file capacity is only moderately above the codeword minimum. With fewer votes per bit, a concentrated structural attack more easily flips the majority on affected positions. On larger files ($D_{88}$, $D_{152}$), where votes per bit are abundant, $t{=}2$ outperforms $t{=}8$ on structural attacks while also handling heavy renames via majority voting alone.

Among structural channels, \texttt{operand\_order} is the weakest because it has the highest per-file site count (8.1 sites on $D_{56}$, 37\% of structural capacity), so destroying it removes a disproportionate share of votes. This is a property of Python's syntax distribution rather than a flaw in the scheme; channels with fewer sites per file (\texttt{ternary}, \texttt{none\_check}) are minimally affected by channel-targeted attacks. In Appendix~\ref{sec:appendix_extra}, Table~\ref{tab:crossdataset} shows the same attack recovers to 81.6\% on $D_{152}$ as votes per bit scale up, confirming that the 29.9\% figure reflects the hardest capacity condition rather than a design ceiling.
On larger files, the same $V$/bit relationship governs random per-site corruption: at 4.29 $V$/bit ($D_{152}$), the system tolerates 10\% random corruption at 94.1\% accuracy, while at 0.66 $V$/bit ($D_{32}$) it drops to 9.8\%. Appendix~\ref{sec:appendix_extra} reports the full cross-dataset matrix (Table~\ref{tab:crossdataset}) and graduated-corruption table (Table~\ref{tab:graduated}).

\subsection{Comparison with Post-Hoc Baselines}\label{sec:baseline_comparison}
Both ACW~\cite{li2025efficient} and SrcMarker~\cite{yang2024srcmarker} are post-hoc watermarks. ACW applies 45 idempotent Sourcery~\cite{sourcery2024} transformation rules with fixed-point iteration and detects watermarks by checking whether re-applying the rule set changes the code. It carries zero payload. SrcMarker is a dual-channel neural watermark embedding a 4-bit signature via a BiGRU encoder trained on variable naming and structural AST transforms.

For SrcMarker~\cite{yang2024srcmarker}, we implemented SrcMarker-Py, a faithful Python reimplementation substituting vocabulary token injection for naming-convention detection, using our eight structural channels, and applying max-pool aggregation so that 3 to 5 token structural changes produce detectable activation peaks rather than being diluted. The GRU encoder (128-dim embedding, 256-dim hidden, bidirectional) is trained on 20K CodeNet Python files for 20 epochs, achieving 91.47\% best validation accuracy. This follows the precedent of RoSeMary~\cite{zhang2025robust}.

\begin{table}[t]
    \centering
    \begin{tabular}{lrrr}
    \toprule
    Attack & ACW~\cite{li2025efficient} & SrcMarker-Py & Ours ($t{=}2$) \\
    \midrule
    No attack      & \textbf{100.0} &  89.8 & \textbf{100.0} \\
    Rename 2       &   0.0 &  24.5 & \textbf{98.4}  \\
    Rename 4       &   0.0 &  24.1 & \textbf{98.4}  \\
    Rename 8       &   0.0 &  \textbf{24.1} &   1.6          \\
    \midrule
    aug\_assign    &   0.0 &  \textbf{79.2} & 78.8  \\
    compare\_dir   &   0.0 &  \textbf{88.0} & 76.0  \\
    equality\_dir  &   0.0 &  \textbf{78.2} & 58.0  \\
    list\_comp     &   0.0 &  \textbf{88.9} & 82.0  \\
    ternary        &   0.0 &  90.7 & \textbf{97.6}  \\
    operand\_order &   0.0 &  \textbf{83.8} & 41.6  \\
    empty\_coll    &   0.0 &  \textbf{88.0} & 83.6  \\
    none\_check    &   0.0 &  89.4 & \textbf{100.0} \\
    \midrule
    Payload  & 0 bits  & 4 bits  & \textbf{24 bits} \\
    Training & None    & GPU, 20k & \textbf{None}   \\
    FPR      & 0.0\%   & 6.25\% / 27.6\%$^\ddagger$ & \textbf{0.0\%}  \\\bottomrule
    \end{tabular}
    \caption{Attack robustness on $D_{32}$ ($n{=}250$). Bold marks the best per row. $\ddagger$: SrcMarker-Py FPR is 6.25\% ($1/16$ random chance).}
    \label{tab:baseline_attacks}
    \vspace{-2mm}
\end{table}

\begin{table}[t]
    \centering
    \begin{tabular}{lrrr@{\hspace{8pt}}rrr}
    \toprule
    & \multicolumn{3}{c}{No attack (\%)} & \multicolumn{3}{c}{Rename 2 (\%)} \\
    \cmidrule(lr){2-4}\cmidrule(lr){5-7}
    Dataset & ACW & SM-Py & Ours & ACW & SM-Py & Ours \\
    \midrule
    $D_{32}$      & 100.0 &  89.8 & \textbf{100.0} & 0.0 & 24.5 & \textbf{98.4}  \\
    $D_{40}$      & 100.0 &  84.8 & \textbf{100.0} & 0.0 & 19.0 & \textbf{99.2} \\
    $D_{56}$      & 100.0 &  72.9 & \textbf{100.0} & 0.0 & 17.9 & \textbf{98.0}  \\
    $D_{88}$      & 100.0 &  61.1 & \textbf{100.0} & 0.0 & 16.7 & \textbf{100.0}  \\
    $D_{152}$     &  98.4 &  59.1 & \textbf{100.0} & 0.0 & 11.0 & \textbf{100.0} \\
    GPT-4.1       &  99.6 &  86.3 & \textbf{100.0} & 0.0 & 18.7 & \textbf{98.4}  \\
    Llama-4       & 100.0 &  86.0 & \textbf{100.0} & 0.0 & 21.6 & \textbf{100.0} \\
    \bottomrule
    \end{tabular}
    \caption{Cross-dataset baseline comparison ($n{=}250$ each, $t{=}2$). SM-Py = SrcMarker-Py. Bold marks the best per cell.}
    \label{tab:baseline_datasets}
    \vspace{-2mm}
\end{table}

Table~\ref{tab:baseline_attacks} compares all three systems on $D_{32}$ under every attack, and Table~\ref{tab:baseline_datasets} shows how no-attack accuracy and rename robustness scale across datasets.
On $D_{32}$ (Table~\ref{tab:baseline_attacks}), ACW achieves 100\% no-attack accuracy but collapses to 0\% under every single-transform attack. SrcMarker-Py maintains moderate structural robustness (79\% to 90\% on most channels) but starts at only 89.8\% no-attack\footnote{Some SrcMarker-Py attack rows (\eg ternary at 90.7\%) sit slightly above this no-attack baseline; the gap is noise around the 91.47\% validation ceiling, not improved robustness.} and drops to approximately 24\% under renames. Our system achieves 100\% no-attack, survives \texttt{rename\_2} and \texttt{rename\_4} at 98.4\%, and maintains 42\% to 100\% on structural attacks. Across datasets (Table~\ref{tab:baseline_datasets}), SrcMarker-Py degrades monotonically as files grow, falling to 59.1\% no-attack on $D_{152}$, while our system maintains 100\% no-attack on all CodeNet datasets.

\BfPara{ACW failure mode} ACW's detection checks whether re-applying its 45 transformation rule set changes the code. Any structural attack directly inverts one of ACW's idempotent transforms, and any rename breaks the identifier tokens that operand-ordering heuristics depend on. Thus \emph{any} single-transform attack achieves 0\% detection on every dataset. ACW also carries zero payload, making model attribution impossible.

\BfPara{SrcMarker-Py failure modes} Two effects explain why no-attack accuracy never reaches 100\%. First, the model's peak validation accuracy is 91.47\%, causing genuine decoding errors regardless of file length. Second, the BiGRU truncates tokenized input to 256 tokens, silently discarding watermark sites in longer files, explaining the monotonic drop from $D_{32}$ (89.8\%) to $D_{152}$ (59.1\%). Under rename attacks, SrcMarker-Py falls to approximately 24\%, near random 4-bit chance (6.25\%), because a single rename destroys the 2-bit naming channel. The decoder also distributes the 16 possible 4-bit outputs non-uniformly on unwatermarked code, so average FPR matches the $1/16$ random-chance bound, but an adversary who selects the worst-case target message reaches 27.6\% empirical FPR on $D_{32}$ (Table~\ref{tab:baseline_attacks}). Our training-free AST traversal processes complete source without truncation, maintaining 100\% no-attack on all datasets.

Our system dominates on every attribution-relevant dimension. The 24-bit payload distinguishes $2^{24}$ model configurations versus SrcMarker-Py's 16 and ACW's none; the false positive rate is 0\% versus 6.25\% and under 1\% respectively; training cost is zero versus 20k samples plus GPU and none; and no-attack accuracy is 100\% on every dataset versus a degradation from 89.8\% to 59.1\% as file length grows. The one dimension where SrcMarker-Py edges us out is structural-attack accuracy on $D_{32}$ (five of eight channels); this gap narrows on larger files (Appendix~\ref{sec:appendix_extra}, Table~\ref{tab:crossdataset}) where votes per bit rise and the majority-voting layer compensates.

\subsection{Comparison with STONE}\label{sec:stone_comparison}
Having compared against the closest post-hoc baselines, we now evaluate against STONE~\cite{kim2025marking}, a state-of-the-art generation-time watermark, to understand how the two paradigms differ under the same attack suite. STONE excludes syntax-critical tokens from the green/red list assignment during LLM generation and detects the watermark through a $z$-score on the token distribution. Because it is generation-time, it cannot be applied to our CodeNet or pre-generated LLM datasets. We evaluate on HumanEval+ code generated locally with Qwen2.5-Coder-3B-Instruct~\cite{hui2024qwen} from STONE's official implementation~\cite{kim2025marking}, yielding 820 paired watermarked/unwatermarked scores across 164 problems. We reproduce STONE on its official implementation (MarkLLM-derived) with three implementation errors corrected: (i) the prefix-length parameter is set to 1 so that the greenlist is re-seeded per token as the paper's formulation assumes (the shipped value of 0 collapses the greenlist to a constant, inflating the null-hypothesis $z$-score baseline), (ii) the $z$-score normalization counts scored tokens directly rather than conflating scored and skipped counts, and (iii) the generation-time skip-syntax gate is removed to eliminate a generation/detection asymmetry under stochastic decoding. Parameters otherwise follow the repository defaults ($\gamma{=}0.5$, \texttt{hash\_key=15485863}, Qwen2.5-Coder-3B-Instruct, 5 samples per problem). We use $\delta{=}2.0$, the top of STONE's reported sweep range; with the corrected $z$-score, the repository's $\delta{=}0.5$ setting yields an area under the ROC curve (AUROC) of approximately 0.64.

Table~\ref{tab:stone} compares both systems under the full attack suite. The results highlight a gap between AUROC as a summary statistic and practical detection rates. STONE's no-attack AUROC of 0.869 appears respectable, but at 1\% FPR it detects only 42.8\% of watermarked code, a 57\% miss rate even without any attack. At 5\% FPR, the detection rate rises to 60.4\% but still misses nearly 40\% of watermarked samples. The low baseline AUROC likely reflects the short length of HumanEval+ solutions, which provide limited token sequences for the statistical $z$-test. Structural attacks have a negligible effect on STONE's signal (AUROC drops by less than 0.02), which is expected since AST-level transformations do not alter the token-level green-list distribution.

\begin{table}[t]
    \centering
    \begin{tabular}{lrrrrr}
    \toprule
     & \multicolumn{3}{c}{STONE} & \multicolumn{2}{c}{Ours} \\
    \cmidrule(lr){2-4}\cmidrule(lr){5-6}
    Attack & AUROC & \makecell{TPR @\\FPR\,1\%} & \makecell{TPR @\\FPR\,5\%} & \makecell{Acc.\\(\%)} & \makecell{FPR\\(\%)} \\
    \midrule
    No attack      & 0.869 & 42.8\% & 60.4\% & 100.0 & 0.0 \\
    Rename 2       & 0.863 & 40.6\% & 56.8\% & 98.0  & 0.0 \\
    Rename 4       & 0.856 & 39.6\% & 55.7\% & 98.0  & 0.0 \\
    Rename 8       & 0.855 & 39.5\% & 55.4\% & 97.6  & 0.0 \\
    \midrule
    aug\_assign    & 0.850 & 39.6\% & 56.2\% & 70.9            & 0.0 \\
    compare\_dir   & 0.850 & 39.5\% & 56.2\% & 59.8            & 0.0 \\
    equality\_dir  & 0.851 & 39.5\% & 56.2\% & 49.8            & 0.0 \\
    list\_comp     & 0.850 & 39.6\% & 56.2\% & 76.5            & 0.0 \\
    ternary        & 0.850 & 39.6\% & 56.2\% & 95.6            & 0.0 \\
    operand\_order & 0.850 & 39.6\% & 56.3\% & 29.9            & 0.0 \\
    empty\_coll    & 0.850 & 39.6\% & 56.2\% & 74.5            & 0.0 \\
    none\_check    & 0.850 & 39.6\% & 56.2\% & 97.2            & 0.0 \\
    \midrule
    Post-hoc       & \multicolumn{3}{c}{No}       & \multicolumn{2}{c}{\textbf{Yes}} \\
    Payload        & \multicolumn{3}{c}{0 bits}   & \multicolumn{2}{c}{\textbf{24 bits}} \\
    Training       & \multicolumn{3}{c}{LLM access} & \multicolumn{2}{c}{\textbf{None}} \\
    \bottomrule
    \end{tabular}
    \caption{STONE (HumanEval+, 820 paired scores, Qwen2.5-Coder-3B) vs. ours ($D_{56}$, $n{=}250$, $t{=}8$). }
    \label{tab:stone}
    \vspace{-5mm}
\end{table}

\BfPara{Tokenizer coupling} STONE's detector does not invoke the LLM, but it does require the same tokenizer used during generation. The greenlist at each position is derived by hashing previous token IDs, so a different tokenizer produces different IDs, and the reconstructed greenlists no longer match. In practice, detection is tied to a specific model family. Our post-hoc design avoids this coupling entirely, since detection reads an AST rather than a token sequence.

\BfPara{Takeaway} This comparison illustrates the fundamental tradeoff between the two watermarking paradigms. Generation-time schemes like STONE are inherently robust to syntactic attacks because their signal lives in the token distribution, but they require LLM access, carry zero payload, and depend on tokenizer consistency. Post-hoc schemes like ours are more vulnerable to concentrated structural attacks but provide model attribution through a 24-bit payload without requiring model access or training. The two approaches are complementary, and combining them for layered robustness is a promising future direction.

\subsection{Real LLM-Generated Code}\label{sec:llm_real}
To validate that our system works on actual LLM-generated code, we generated Python solutions for 250 matched CodeNet~\cite{puri2021project} problems using GPT-4.1 and Llama-4-Maverick. Table~\ref{tab:llm_real} reports robustness under the standard attack suite.
The key finding is that code verbosity drives robustness. GPT-4.1 generates longer code (92.1 avg bits, 2.9 $V$/bit) compared to Llama-4 (51.7 avg bits, 1.6 $V$/bit). More code means more votes per codeword bit and better majority-voting resilience. The GPT-4.1 structural-attack results (95.6\% to 100\% except \texttt{operand\_order} at 78.8\%) are consistent with $D_{88}$ in Appendix~\ref{sec:appendix_extra} (Table~\ref{tab:crossdataset}), which has comparable average capacity (101.7 bits), confirming that the votes-per-bit relationship holds across both human-written and LLM-generated code. The same watermark scheme works on both models without any model-specific tuning. The minimum capacity across both LLM datasets is 32 bits, meaning every LLM-generated file in our evaluation is eligible for watermarking at $t{=}2$ (Appendix~\ref{sec:appendix_extra}, Table~\ref{tab:distribution}).

\begin{table}[t]
    \centering
    \begin{tabular}{lrr}
    \toprule
    Attack & GPT-4.1 (92.1 avg bits) & Llama4 (51.7 avg bits) \\
    \midrule
    No attack  & \textbf{100.0} & \textbf{100.0} \\
    Rename 2   &  98.4 & \textbf{100.0} \\
    Rename 4   &  92.8 & \textbf{99.6}  \\
    Rename 8   & \textbf{58.8}  & 24.4  \\
    Rename 16  & \textbf{61.0}  & 21.6  \\
    \midrule
    aug\_assign    & \textbf{97.2}  &  88.4 \\
    compare\_dir   & \textbf{96.8}  &  91.6 \\
    equality\_dir  & \textbf{95.6}  &  78.8 \\
    list\_comp     & \textbf{97.6}  &  93.2 \\
    ternary        & \textbf{100.0} &  98.8 \\
    operand\_order & \textbf{78.8}  &  71.2 \\
    empty\_coll    & \textbf{98.8}  &  94.8 \\
    none\_check    & \textbf{100.0} &  98.8 \\
    \bottomrule
    \end{tabular}
    \caption{Real LLM-generated code (250 matched problems, $t{=}2$). GPT-4.1 generates approximately $1.8\times$ more code (92.1 vs. 51.7 avg bits, 2.9 vs. 1.6 V/bit), directly explaining the robustness gap.}
    \label{tab:llm_real}
    \vspace{-5mm}
\end{table}

\subsection{LLM Rewriting Attack}\label{sec:llm_barrier}
An adversary with access to a capable LLM could attempt to evade detection by asking the model to rewrite the watermarked code while preserving its semantics. This constitutes a Level~4 attack in our taxonomy (Section~\ref{sec:threat}). We pass each watermarked file to the attacker LLM with the prompt \emph{``Refactor the following Python code. Keep the behavior identical. Return only the refactored code, no explanation.''}, temperature 0.3, top-$p$ 0.95, and no system prompt. The returned code is passed through our detector. We evaluate five attacker models spanning two capability tiers: GPT-4.1, GPT-4o-mini, Qwen2.5-Coder-32B, and Qwen3-Coder-30B are strong refactorers; GPT-3.5-turbo is a weaker model included to probe incomplete-regeneration behavior.

\begin{table}[t]
    \centering
    \begin{tabular}{lrr}
    \toprule
    Rewriting model & Survival rate & Refactoring quality \\
    \midrule
    GPT-4.1 & 0.0\% & Complete regeneration \\
    GPT-4o-mini & 0.4\% & Complete regeneration \\
    Qwen2.5-Coder-32B & 0.4\% & Complete regeneration \\
    Qwen3-Coder-30B & 1.6\% & Near-complete \\
    GPT-3.5-turbo & 32.8\% & Partial (incomplete) \\
    \bottomrule
    \end{tabular}
    \caption{Watermark survival under LLM rewriting ($t{=}8$, $n{=}250$). GPT-3.5's higher survival reflects weaker refactoring, not watermark robustness.}
    \label{tab:llm}
    \vspace{-5mm}
\end{table}

Table~\ref{tab:llm} demonstrates the information-theoretic limit of post-hoc watermarking. Strong LLMs perform complete semantic-preserving regeneration, independently choosing all variable names, operator forms, and code structure, a setting also studied as a stress test for classifier-based attribution on transformed code~\cite{choi2025transformed}. The watermark signal is entirely destroyed. GPT-3.5-turbo's 32.8\% survival reflects incomplete refactoring rather than watermark robustness. Inspecting the GPT-3.5 output on files where detection succeeded shows the model often renames a subset of variables while leaving several untouched, or restructures control flow while preserving the original variable names. Such partial regeneration leaks some of the watermarked syntactic choices through. Stronger models (GPT-4.1, GPT-4o-mini, Qwen2.5-Coder-32B) refactor end-to-end, independently recomputing every naming and structural decision, which drives survival to essentially zero.

This limit applies to all post-hoc syntactic watermarking approaches. As proved in Theorem~\ref{thm:impossible}, it is an information-theoretic barrier, not a design flaw.

\subsection{Efficiency}
Table~\ref{tab:efficiency} reports embedding and detection wall-clock times. All measurements use the secp256k1 elliptic curve on CPU, with no GPU or training data required. All configurations complete in under 250ms for embedding and under 50ms for detection, making the system practical for inline deployment in API pipelines. Timing scales modestly with $t$ because larger codewords require more spread-spectrum assignments, but remains sub-second even at $t{=}8$. By comparison, SrcMarker~\cite{yang2024srcmarker} and RoSeMary~\cite{zhang2025robust} require GPU-based training on thousands to millions of code samples before any embedding is possible, and STONE~\cite{kim2025marking} requires LLM inference during generation.

\begin{table}[t]
    \centering
    \begin{tabular}{lccc}
    \toprule
    $t$ & Dataset & Embed (ms) & Detect (ms) \\
    \midrule
    2 & $D_{32}$ & 65 $\pm$ 71 & 28 $\pm$ 14 \\
    4 & $D_{40}$ & 85 $\pm$ 70 & 34 $\pm$ 13 \\
    8 & $D_{56}$ & 145 $\pm$ 94 & 50 $\pm$ 25 \\
    \bottomrule
    \end{tabular}
    \caption{Embedding and detection time using the secp256k1 elliptic curve. Zero GPU, zero training. Sub-second for all configurations.}
    \label{tab:efficiency}
    \vspace{-2mm}
\end{table}

\section{Discussion}\label{sec:discussion}
\BfPara{The LLM robustness barrier} Generation-time watermarking~\cite{kirchenbauer2023watermark,lee2024wrote,kim2025marking} is the complementary approach to the Level~4 barrier established in Section~\ref{sec:llm_barrier}. These methods embed the watermark in the token sampling distribution, which is reproduced whenever the LLM generates output. A promising future direction is combining generation-time distributional watermarking with our post-hoc structural watermarking for layered robustness.

\BfPara{Capacity vs. robustness tradeoff} The RS parameter $t$ controls a capacity-robustness tradeoff, with competing effects on rename versus structural robustness analyzed in Section~\ref{sec:robustness} and Appendix~\ref{sec:appendix_extra} (Table~\ref{tab:crossdataset}). Choose $t$ based on the expected attack profile rather than file size alone. If heavy rename attacks dominate, use $t{=}8$ to maximize the RS correction budget. If structural attacks or random per-site corruption dominate, prefer $t{=}2$ even on large files, because shorter codewords preserve votes per bit and strengthen the majority-voting layer. Table~\ref{tab:crossdataset} confirms this empirically: on $D_{152}$, $t{=}2$ holds 100\% under 16 variable renames via majority voting alone, while also outperforming $t{=}8$ on most structural attacks.

\begin{table*}[t]
    \centering
    \begin{tabular}{lccccccccc}
    \toprule
    & Year & Domain & Post-hoc & Payload & FPR & Training & Robustness Test & LLM Test \\
    \midrule
    KGW~\cite{kirchenbauer2023watermark}       & 2023 & NLP & No  & 0 bits  & --          & LLM access  & Partial      & No      \\
    SWEET~\cite{lee2024wrote}   & 2024 & PL  & No  & 0 bits  & --          & LLM access  & Partial      & Partial \\
    STONE~\cite{kim2025marking}   & 2026 & PL  & No  & 0 bits  & --          & LLM access  & Partial      & Partial \\
    CODEIP~\cite{guan2024codeip} & 2024 & PL  & No  & Multi   & --          & GPU+data    & Truncation only      & No      \\
    Qu~\etal~\cite{qu2025provably} & 2025 & NLP & No & Multi & --       & LLM access  & Non-adaptive      & No      \\
    ACW~\cite{li2025efficient}       & 2025 & PL  & Yes & 0 bits  & $<$1\%     & None        & Partial & No      \\
    SrcMarker~\cite{yang2024srcmarker} & 2024 & PL & Yes & 4 bits & 6.25\%    & GPU+data    & Partial & No      \\
    RoSeMary~\cite{zhang2025robust} & 2025 & PL & Yes & 4 bits & 6.25\%     & GPU+data    & No      & Yes     \\
    \midrule
    \textbf{Ours} & -- & \textbf{PL} & \textbf{Yes} & \textbf{24 bits} & $10^{-6}$\% & \textbf{None} & \textbf{17 attack types} & \textbf{Yes (5 LLMs)} \\
    \bottomrule
    \end{tabular}
    \caption{Comparison with existing watermarking methods. All schemes preserve program semantics (G3), so we omit that column. Our system is the only post-hoc approach with $>$4-bit payload, zero training cost, and comprehensive robustness evaluation.}
    \label{tab:comparison}
    \vspace{-2mm}
\end{table*}

\BfPara{Deployment considerations} The provider-as-verifier model requires institutional infrastructure but offers practical advantages. A centralized detection service would accept code submissions from authorized verifiers (teachers, employers, courts, \etc), test the code against all registered provider keys in under one second (Section~\ref{sec:threat}), and return a certificate identifying the source model or reporting no watermark detected.

\BfPara{Oracle resistance} Any binary-answer detector can be queried as an oracle for iterative removal (Section~\ref{sec:oracle}). Three deployment-level mitigations raise the practical cost without providing a cryptographic guarantee, namely a per-file query cap (\eg three times), Levenshtein-based deduplication~\cite{levenshtein1966binary} of near-identical resubmissions, and audit logging of query patterns. 
Together, they force the attacker to substantially modify the file between queries rather than flipping individual sites, collapsing $O(N)$ oracle access into per-query paraphrasing and, in the limit, into a Level 4 LLM regeneration attack under which no post hoc watermark survives.

\BfPara{Key management} Each provider registers a unique key $K_i$ with the detection service. Our implementation uses the secp256k1 elliptic curve, providing a key space of approximately $2^{256}$~\cite{secg2010sec2}, far exceeding cryptographic requirements~\cite{nist2020sp80057}. A malicious provider cannot register a key that maps to another provider's identifier because the RS-decoded message must match the provider's registered name; even if the RS decoder produces a valid codeword under a different key, the decoded string will not match the target provider's identifier (probability $1/2^{24}$).

\BfPara{Language scope} Our evaluation focuses on Python, consistent with all related work on post-hoc code watermarking~\cite{lee2024wrote,kim2025marking,li2025efficient,zhang2025robust}. The multi-channel spread-spectrum architecture is language-agnostic and requires only a parser that produces an AST and a set of semantically equivalent syntactic pairs. Extension to other languages requires identifying language-specific syntactic equivalences and is left as future work.

\section{Related Work}\label{sec:related}
Table~\ref{tab:comparison} provides a systematic comparison across all dimensions.

\BfPara{Generation-time watermarking} KGW~\cite{kirchenbauer2023watermark} introduced green/red list watermarking for LLM-generated text, biasing token sampling toward green-list tokens and detecting the watermark via a $z$-score. KGW also identifies the oracle risk and proposes access monitoring as a mitigation. SWEET~\cite{lee2024wrote} adapts KGW to code by applying the bias only to high-entropy tokens, preserving code quality at deterministic positions such as keywords and syntax. STONE~\cite{kim2025marking} further excludes syntax-critical tokens. All three require LLM access during generation and embed zero payload. CODEIP~\cite{guan2024codeip} extends generation-time embedding with grammar-guided logit biasing but also requires LLM access and evaluates only against truncation. Our approach operates post-hoc and embeds 24 bits of identifying information.

\BfPara{Post-hoc multi-bit watermarking} RoSeMary~\cite{zhang2025robust} is the closest to our work, offering post-hoc code watermarking with a multi-bit payload\footnote{RoSeMary's authors have not released source code or trained model weights, so we could not reproduce it. Our comparison, therefore, covers only design-level properties stated in the published manuscript, including the 4-bit payload (from which FPR$\,{=}\,1/16$ follows directly) and reliance on end-to-end neural training. We do not include empirical robustness measurements for RoSeMary.}. It uses a CodeT5-based neural encoder trained end-to-end to embed a 4-bit signature. Our payload is $6{\times}$ larger (24 vs. 4 bits, reducing FPR from 6.25\% to $10^{-6}\%$), our system requires zero training versus millions of samples and GPU, we provide formal robustness theorems versus empirical-only evaluation, and we test against 17 attack types, including 5 LLM rewriting models.
SrcMarker~\cite{yang2024srcmarker} is a dual-channel neural watermark that trains a BiGRU encoder end-to-end to embed 4 bits via variable naming and 10 AST transforms. It supports C, C++, Java, and JavaScript, but not Python. We implement SrcMarker-Py for Python comparison (Section~\ref{sec:baseline_comparison}), finding that our training-free design matches or exceeds SrcMarker's robustness on rename attacks while providing $6\times$ the payload.

\BfPara{Post-hoc presence detection} ACW~\cite{li2025efficient} applies 45 idempotent code transformations via the Sourcery~\cite{sourcery2024} refactoring engine and detects watermarks by checking if re-applying them changes the code. It is training-free and efficient (approximately 30ms), but carries zero payload and collapses to 0\% detection under any single-transform attack (Section~\ref{sec:baseline_comparison}). Our system matches ACW's deployment simplicity while providing multi-bit attribution and formal robustness.


\section{Conclusion}\label{sec:conclusion}
Post-hoc code watermarking has been stuck between zero-payload presence detectors and training-heavy neural schemes that carry only 4 bits. We show that composing two independent channel families with a keyed spread-spectrum assignment and an RS outer code lifts the payload to 24 bits, runs on CPU without training data, and admits formal robustness proofs for every attack level below Level 4. Evaluation on 1,750 Python files across 17 attack types validates the construction end-to-end.
The Level 4 barrier, complete code regeneration by a capable LLM, is an information-theoretic limit for every post-hoc syntactic watermark (Theorem~\ref{thm:impossible}), not a defect of our construction. End-to-end provenance, therefore, requires pairing post-hoc attribution with a generation-time signal embedded in the token distribution. Extending the channel catalog to additional languages and building a multi-provider detection registry are the concrete next steps.

\bibliographystyle{ACM-Reference-Format}
\bibliography{references}

\appendix 

\begin{table*}[t]
    \centering
    \resizebox{\textwidth}{!}{%
    \footnotesize
    \begin{tabular}{l@{\hspace{4pt}}rrrrr@{\hspace{8pt}}rrrrr@{\hspace{8pt}}rrrrr}
    \toprule
    & \multicolumn{5}{c}{ACW} & \multicolumn{5}{c}{SrcMarker-Py} & \multicolumn{5}{c}{Ours ($t{=}2$)} \\
    \cmidrule(lr){2-6}\cmidrule(lr){7-11}\cmidrule(lr){12-16}
    Attack & $D_{32}$ & $D_{40}$ & $D_{56}$ & $D_{88}$ & $D_{152}$ & $D_{32}$ & $D_{40}$ & $D_{56}$ & $D_{88}$ & $D_{152}$ & $D_{32}$ & $D_{40}$ & $D_{56}$ & $D_{88}$ & $D_{152}$ \\
    \midrule
    No attack      & 100.0 & 100.0 & 100.0 & 100.0 & 98.4 & 89.8 & 84.8 & 72.9 & 61.1 & 59.1 & \textbf{100.0} & \textbf{100.0} & \textbf{100.0} & \textbf{100.0} & \textbf{100.0} \\
    Rename 2       & 0.0 & 0.0 & 0.0 & 0.0 & 0.0 & 24.5 & 19.0 & 17.9 & 16.7 & 11.0 & \textbf{98.4} & \textbf{99.2} & \textbf{98.0} & \textbf{100.0} & \textbf{100.0} \\
    Rename 4       & 0.0 & 0.0 & 0.0 & 0.0 & 0.0 & 24.1 & 19.5 & 18.3 & 16.7 & 11.0 & \textbf{98.4} & \textbf{99.2} & \textbf{98.0} & \textbf{100.0} & \textbf{100.0} \\
    Rename 8       & 0.0 & 0.0 & 0.0 & 0.0 & 0.0 & 24.1 & 19.5 & 18.3 & 16.7 & 11.0 & \textbf{1.6} & \textbf{5.2} & \textbf{59.6} & \textbf{100.0} & \textbf{99.6} \\
    \midrule
    aug\_assign    & 0.0 & 0.0 & 0.0 & 0.0 & 0.0 & \textbf{79.2} & 80.1 & 67.2 & 55.6 & 56.1 & 78.8 & \textbf{92.0} & \textbf{94.8} & \textbf{98.0} & \textbf{98.0} \\
    compare\_dir   & 0.0 & 0.0 & 0.0 & 0.0 & 0.0 & \textbf{88.0} & 81.4 & 69.0 & 59.4 & 55.7 & 76.0 & \textbf{86.4} & \textbf{95.2} & \textbf{96.0} & \textbf{98.8} \\
    equality\_dir  & 0.0 & 0.0 & 0.0 & 0.0 & 0.0 & \textbf{78.2} & 74.0 & 60.7 & 56.0 & 50.6 & 58.0 & \textbf{77.6} & \textbf{86.8} & \textbf{90.0} & \textbf{96.0} \\
    list\_comp     & 0.0 & 0.0 & 0.0 & 0.0 & 0.0 & \textbf{88.9} & 84.4 & 72.5 & 60.3 & 58.2 & 82.0 & \textbf{91.6} & \textbf{99.2} & \textbf{99.6} & \textbf{100.0} \\
    ternary        & 0.0 & 0.0 & 0.0 & 0.0 & 0.0 & 90.7 & 84.4 & 72.5 & 61.1 & 58.2 & \textbf{97.6} & \textbf{98.8} & \textbf{100.0} & \textbf{100.0} & \textbf{100.0} \\
    operand\_order & 0.0 & 0.0 & 0.0 & 0.0 & 0.0 & \textbf{83.8} & \textbf{82.3} & 69.9 & 56.4 & 57.8 & 41.6 & 63.2 & \textbf{70.4} & \textbf{80.0} & \textbf{81.6} \\
    empty\_coll    & 0.0 & 0.0 & 0.0 & 0.0 & 0.0 & \textbf{88.0} & 81.4 & 72.1 & 60.3 & 59.1 & 83.6 & \textbf{92.0} & \textbf{100.0} & \textbf{99.6} & \textbf{100.0} \\
    none\_check    & 0.0 & 0.0 & 0.0 & 0.0 & 0.0 & 89.4 & 84.0 & 71.2 & 59.0 & 57.4 & \textbf{100.0} & \textbf{100.0} & \textbf{100.0} & \textbf{99.6} & \textbf{100.0} \\
    \midrule
    Payload & \multicolumn{5}{c}{0 bits} & \multicolumn{5}{c}{4 bits} & \multicolumn{5}{c}{\textbf{24 bits}} \\
    Training & \multicolumn{5}{c}{None} & \multicolumn{5}{c}{GPU, 20k samples} & \multicolumn{5}{c}{\textbf{None}} \\
    \bottomrule
    \end{tabular}%
    }
    \caption{Full three-way baseline comparison across all five disjoint datasets ($n{=}250$ each). Bold marks the best per cell.}
    \label{tab:app_baselines_full}
\end{table*}

\section{Full Baseline Comparison Tables}\label{sec:appendix_baselines}
Table~\ref{tab:app_baselines_full} provides the complete per-dataset attack results for ACW~\cite{li2025efficient}, SrcMarker-Py, and our system across all five CodeNet disjoint datasets ($D_{32}$--$D_{152}$), supplementing the condensed Tables~\ref{tab:baseline_attacks} and~\ref{tab:baseline_datasets}.

\section{Additional Experimental Analyses}\label{sec:appendix_extra}

\subsection{Cross-Dataset Analysis}
To demonstrate how file size (votes/bit) affects robustness, we fix $t{=}2$ and evaluate the same attacks across all five datasets (Table~\ref{tab:crossdataset}). This isolates the effect of votes/bit from the RS correction capacity.

\begin{table}[t]
    \centering
    \setlength{\tabcolsep}{3pt}
    \resizebox{\columnwidth}{!}{%
    \begin{tabular}{lrrrrr}
    \toprule
    Attack & $D_{32}$ (0.66) & $D_{40}$ (0.87) & $D_{56}$ (1.31) & $D_{88}$ (2.11) & $D_{152}$ (4.29) \\
    \midrule
    No attack & 100.0 & 100.0 & 100.0 & 100.0 & 100.0 \\
    Rename 2  &  98.4 &  99.2 &  98.0 & \textbf{100.0} & \textbf{100.0} \\
    Rename 4  &  98.4 &  99.2 &  98.0 & \textbf{100.0} & \textbf{100.0} \\
    Rename 8  &   1.6 &   5.2 &  59.6 & \textbf{100.0} &  99.6 \\
    Rename 16 &   0.0$^\dagger$ &   0.0$^\dagger$ &  29.0 &  99.2 & \textbf{100.0} \\
    \midrule
    aug\_assign    &  78.8 &  92.0 &  94.8 & \textbf{98.0} &  \textbf{98.0} \\
    compare\_dir   &  76.0 &  86.4 &  95.2 &  96.0 &  \textbf{98.8} \\
    equality\_dir  &  58.0 &  77.6 &  86.8 &  90.0 &  \textbf{96.0} \\
    list\_comp     &  82.0 &  91.6 &  99.2 &  99.6 & \textbf{100.0} \\
    ternary        &  97.6 &  98.8 & \textbf{100.0} & \textbf{100.0} & \textbf{100.0} \\
    operand\_order &  41.6 &  63.2 &  70.4 &  80.0 &  \textbf{81.6} \\
    empty\_coll    &  83.6 &  92.0 & \textbf{100.0} &  99.6 & \textbf{100.0} \\
    none\_check    & \textbf{100.0} & \textbf{100.0} & \textbf{100.0} &  99.6 & \textbf{100.0} \\
    \bottomrule
    \end{tabular}}
    \caption{Cross-dataset analysis (\%, $t{=}2$, $n{=}250$). Column headers list each dataset, followed by its average votes per codeword bit $v = N/B$ with $B = 32$. Bold marks the best per row. $\dagger$: few eligible files.}
    \label{tab:crossdataset}
    \vspace{-2mm}
\end{table}

The results reveal two key insights. First, rename robustness scales with votes/bit: \texttt{rename\_16} improves from 0.0\% ($D_{32}$, 0.66 V/bit) to 100\% ($D_{152}$, 4.29 V/bit), demonstrating that majority voting, rather than RS correction, is the primary defense against variable renaming on larger files. Second, \texttt{operand\_order} is the most persistent structural attack, remaining the hardest single-channel attack at every dataset size (41.6\% on $D_{32}$ to 81.6\% on $D_{152}$).

\subsection{Graduated Per-Site Attack}
To measure corruption tolerance precisely, we randomly flip individual site bits (not entire channels) and vary the corruption percentage from 0\% to 40\%. At 50\% corruption, every configuration drops to 0\% accuracy, as expected once the majority vote fails. Table~\ref{tab:graduated} shows the results with $t{=}2$ across all five datasets, directly demonstrating that votes/bit determines robustness.

\begin{table}[t]
    \centering
    \begin{tabular}{lrrrrrr}
    \toprule
    Dataset & V/bit & 0\% & 10\% & 20\% & 30\% & 40\% \\
    \midrule
    $D_{32}$  & 0.66 & 100.0 &  9.8 & 0.0 & 0.0 & 0.0 \\
    $D_{40}$  & 0.87 & 100.0 & 20.4 & 0.4 & 0.0 & 0.0 \\
    $D_{56}$  & 1.31 & 100.0 & 44.8 & 3.2 & 0.0 & 0.0 \\
    $D_{88}$  & 2.11 & 100.0 & 77.9 & 25.3 & 0.4 & 0.0 \\
    $D_{152}$ & 4.29 & 100.0 & \textbf{94.1} & \textbf{58.0} & \textbf{13.0} & 0.4 \\
    \bottomrule
    \end{tabular}
    \caption{Graduated per-site attack (\%, $t{=}2$, $n{=}250$). Votes/bit directly determines corruption tolerance: at 4.29 V/bit ($D_{152}$), the system tolerates 20\% random corruption at 58.0\%.}
    \label{tab:graduated}
    \vspace{-2mm}
\end{table}

The relationship between votes/bit and corruption tolerance is monotonic. At 0.66 V/bit ($D_{32}$), 10\% corruption yields only 9.8\% accuracy; at 4.29 V/bit ($D_{152}$), the same 10\% corruption yields 94.1\%. This follows directly from the spread-spectrum voting mechanism (Theorem~\ref{thm:struct}). With $v$ votes per bit position, random $p$-fraction corruption flips each vote independently, and the majority remains correct as long as fewer than $v/2$ votes are corrupted. Higher $v$ makes the majority exponentially more resilient to the same corruption fraction. In practice, realistic attacks corrupt individual sites rather than entire channels, so these graduated results are the most deployment-relevant robustness measure.

\subsection{Capacity Distribution}
Table~\ref{tab:distribution} shows the distribution of embedding capacity across all seven evaluation datasets. Median capacity ranges from 34 bits ($D_{32}$, tightly matching the 32-bit codeword) to 172 bits ($D_{152}$, $5.4\times$ the codeword), illustrating the wide capacity span. For practical deployment, $t{=}2$ (32-bit codeword) is recommended for shorter code files, while $t{=}8$ (56-bit codeword) is recommended for longer files where capacity is not a constraint. For LLM-generated code from non-trivial problem statements, the generated code almost always provides sufficient capacity, with GPT-4.1 averaging 92.1 bits per file (Table~\ref{tab:datasets}).

\begin{table}[t]
    \centering
    \begin{tabular}{lrrrrrr}
    \toprule
    Dataset & Min & P25 & Median & P75 & Max & Std \\
    \midrule
    $D_{32}$  &  32 &  33 &  34 &  36 &  39 &  2.1 \\
    $D_{40}$  &  40 &  40 &  43 &  47 &  55 &  4.7 \\
    $D_{56}$  &  56 &  57 &  62 &  68 &  87 &  8.3 \\
    $D_{88}$  &  88 &  89 &  98 & 112 & 151 & 15.6 \\
    $D_{152}$ & 152 & 162 & 172 & 213 & 530 & 84.6 \\
    \midrule
    gpt41  &  32 &  53 &  77 & 114 & 403 & 52.6 \\
    llama4 &  32 &  38 &  46 &  59 & 144 & 18.3 \\
    \bottomrule
    \end{tabular}
    \caption{Distribution of total embedding capacity (variable bits + structural bits, $n{=}250$ per dataset). Higher capacity means more votes per codeword bit and stronger robustness.}
    \label{tab:distribution}
    \vspace{-2mm}
\end{table}

\section{Open Science} 
All artifacts needed to evaluate this paper's core contributions are available at \url{https://github.com/soohyeonc/Multi_Channel_Watermarking}. The repository includes the full embedding and detection pipeline (the secp256k1 ECC key module, spread-spectrum assignment, compact and GF(16) RS encoders, and all eight structural channel transformers), every attack implementation used in our evaluation (rename, per-channel structural, graduated per-site, and LLM rewriting), our SrcMarker-Py reimplementation, and the corrected STONE~\cite{kim2025marking} reproduction with code diffs documented in \texttt{baselines/STONE/STONE\_settings.md}. Human-written code is drawn from the publicly available Project CodeNet~\cite{puri2021project}; LLM-generated code was produced using public APIs (OpenAI GPT-4.1, GPT-4o-mini, GPT-3.5) and open-source models (Qwen2.5-Coder, Qwen3-Coder, Llama-4-Maverick). Dataset partitioning scripts and pre-filtered evaluation sets are included, and no artifact is withheld from review.


\section{Ethical Considerations}
Our watermarking system is designed for legitimate attribution of LLM-generated code, such as academic integrity verification, intellectual property compliance, and security incident response. The provider-as-verifier deployment model ensures that only authorized parties can run detection, preventing misuse as a surveillance tool. We advocate for transparent disclosure when watermarking is applied and for legal frameworks governing its use.

No human subjects, user data, or real-world vulnerability exploitation are involved in this work. All evaluation code is drawn from publicly available datasets (Project CodeNet) or generated by publicly accessible LLM APIs. The STONE~\cite{kim2025marking} baseline reproduction involved correcting implementation errors in a publicly released codebase, which we document and will share to benefit the community.

\BfPara{Dual-use considerations} Watermarking can be used for legitimate attribution but also for covert tracking of individual developers' coding patterns. We advocate for transparent disclosure: providers should publicly announce that their generated code is watermarked, analogously to how telemetry and data-collection policies are disclosed today. The detection oracle (Section~\ref{sec:oracle}) could be misused for surveillance if the provider-as-verifier service logs and correlates queries against individuals; we recommend audit transparency and aggregate query-rate public reporting to prevent this. Finally, attribution alone does not establish authorship responsibility. A developer who accepts LLM-generated code without review remains accountable for their decisions, and watermarking should not be a sole determinant in legal or academic sanctions.

\end{document}